%% file: main.tex
\pdfoutput=1

\newif\ifdraft \draftfalse
\newif\iffull \fulltrue

\makeatletter \@input{tex.flags} \makeatother

\iffull
\documentclass[11pt]{article}
\usepackage{fullpage}
\else
\documentclass{llncs}
\fi

\input{prelude}

\begin{document}

\title{Computer-aided Verification in Mechanism
  Design}
\author{
Gilles Barthe
\and
Marco Gaboardi
\and
Emilio Jes\'us Gallego Arias
\and
Justin Hsu
\and
Aaron Roth
\and
Pierre-Yves Strub
}

\maketitle

\begin{abstract}
  \iffull
In mechanism design, the gold standard solution concepts are \emph{dominant
  strategy incentive compatibility} and \emph{Bayesian incentive compatibility}.
These solution concepts relieve the (possibly unsophisticated) bidders from the
need to engage in complicated strategizing.  While incentive properties are
simple to state, their proofs are specific to the mechanism and can be quite
complex. This raises two concerns.  From a practical perspective, checking a
complex proof can be a tedious process, often requiring experts knowledgeable in
mechanism design. Furthermore, from a modeling perspective, if unsophisticated
agents are unconvinced of incentive properties, they may strategize in
unpredictable ways.
\fi

\iffull To address both concerns, we \fi
explore techniques from computer-aided verification to construct formal proofs
of incentive properties. Because formal proofs can be automatically checked,
agents do not need to manually check the properties, or even understand the
proof. To demonstrate, we present the verification of a sophisticated mechanism:
the generic reduction from Bayesian incentive compatible mechanism design to
algorithm design given by Hartline, Kleinberg, and Malekian. This mechanism
presents new challenges for formal verification, including essential use of
randomness from both the execution of the mechanism and from the prior type
distributions.
\iffull As an immediate consequence, our work also formalizes Bayesian incentive
compatibility for the entire family of mechanisms derived via this reduction.
Finally, as an intermediate step in our formalization, we provide the first
formal verification of incentive compatibility for the celebrated
Vickrey-Clarke-Groves mechanism. \fi
\end{abstract}







\begin{bibunit}[abbrvnat]

\section{Introduction}
\label{sec:intro}

Recent years have seen a surge of interest in mechanism design, as researchers
explore connections between computer science and economics. This fruitful
collaboration has produced many sophisticated mechanisms, including mechanism
deployed in high-stakes auctions. Many mechanisms satisfy properties that
\emph{incentivize} agents to behave in a straightforward and easily modeled
manner; the gold standard properties are \emph{dominant strategy truthful} (in
settings of complete information) and \emph{Bayesian incentive compatible} (in
settings of incomplete information).  While existing mechanisms are impressive
achievements, their increasing complexity raises two concerns.

The first concern is correctness. As mechanisms become more sophisticated,
proofs of their incentive properties have also grown in complexity, sometimes
involving delicate reasoning about randomization or tedious case analysis.
Complex mechanisms are also more prone to implementation errors. The second
concern is more subtle.  At its heart, mechanism design is algorithm design
together with a predictive model of how agents will decide to behave. Unlike
algorithm design, where correctness can be verified in a vacuum, the success of
a mechanism requires a realistic behavioral model of the participants. How will
agents behave when faced with a complex mechanism?

Different behavioral models assume different answers to this question. At one
extreme, we may assume that agents will coordinate to play a Nash equilibrium of
the game and we can study concepts like the \emph{price of anarchy}
\citep{Rou05,CK05}. However, Nash equilibria are generally not unique, requiring
coordination and communication to achieve \citep{HM07}. Even if
information is centralized, equilibria can be computationally hard to find
\citep{DGP09}. Assuming that agents play at a Nash equilibrium may be
unrealistic unless agents possess strong computational resources.

At the other extreme, we may ask for mechanisms which are dominant strategy
truthful or Bayesian incentive compatible. In such mechanisms, agents can do no
better than truthfully reporting their type, even in the worst case or in
expectation over the other agents' types. These solution concepts assume little
about the bidders: When interacting with truthful mechanisms, agents do not have
to engage in complicated counter-speculation, communication, or
computation---they merely have to tell the truth!

However, even with mechanisms that are dominant strategy truthful or Bayesian
incentive compatible, participating agents must still \emph{believe} that the
mechanism is truthful. For complicated mechanisms this is no small matter, as
the incentive properties may require significant domain expertise to verify.  We
are not the first to raise these concerns. \citet{Li2015} argued for simplicity
as a desired feature of auctions, proposing a formal definition; when designing
the FCC auction for reallocating radio spectrum, \citet{milgromslides} advocated
an ``\emph{obviously} strategy-proof'' mechanism.

However, some useful mechanisms are just too complex to be obvious.
\citet{gross2015common}, reporting on experiences with the Denver and New
Orleans school choice system, describe the problem:
\begin{quote}
  Both Denver and New Orleans leaders aggressively conveyed the optimal choosing
  strategy to parents, and many of the parents we spoke with had received the
  message. Parents reported to us that they were told to provide the full number
  of choices in their true order of preference. The problem was that few parents
  actually trusted this message. Instead, they commonly pursued strategies that
  matched their own inaccurate explanations of how the match worked.
\end{quote}
\citet{hassidim2015strategic} report similar behavior in a system for matching
Psychology students to slots in graduate programs in Israel. Even though the
mechanism is dominant-strategy incentive compatible, nearly $20\%$ of applicants
obviously misrepresented their preferences, with possibly more applicants
manipulating their reports in more subtle ways.  Instead of restricting
mechanisms, can we give users evidence for the incentive properties?

In this work, we consider using \emph{formal proofs} as certificates.  Formal
proofs bear a resemblance to pen-and-paper proofs, but they are constructed in a
rigorous fashion: They use a formal syntax, have a precise interpretation as a
mathematical proof, and can be built with a rich palette of computer-assisted
proof-construction tools. Compared to pen-and-paper proofs, the major benefit
of formal proofs is that once constructed, they can be checked independently and
fully automatically by a \emph{proof checker} program. 

Several previous works have explored formal methods for verifying mechanisms;
\citet{Kerber0R16} provide an extensive survey. Broadly speaking, prior work
falls into two groups. \emph{Automated} approaches check properties via
extensive search, guided by intelligent heuristics. These techniques are more
suited to verifying simpler properties of mechanisms, perhaps instantiated on a
specific input; properties like BIC lie beyond the reach of existing approaches.

More manual (sometimes called \emph{interactive}) techniques divide the
verification task into two separate stages. In the first stage, the formal proof
is \emph{constructed}. This step typically involves human assistance, perhaps
encoding the mechanism in a specific form or constructing a formal proof. With
the help of the human, these techniques can prove rich properties like BIC and
support the level of generality that is typical of existing proofs---say, for an
arbitrary number of agents, or for any type space. In the second stage, the
formal proof is \emph{checked}. This step proceeds fully automatically: a proof
checking program verifies that the formal proof is constructed correctly.  This
neat division of the verification task is a natural fit for mechanism design. We
could imagine that the mechanism designer---a sophisticated party who is
intimately familiar with the details of the proof---has the resources and
knowledge to construct the formal proof.  This proof could then be transmitted
to the agents, who can automatically check the proof with no knowledge of the
details.  

The main difference between manual techniques is in the amount of human labor
for proof construction, the most challenging phase. Existing verification
approaches formalize the proof at a level that is far more detailed than
existing proofs on paper, requiring extensive expertise in formal methods.
Furthermore, existing works focus on general correctness properties---the output
of a mechanism should be a partition, the prices should be non-negative, etc.,
rather than incentive properties.

In our work, we look to combine the best of both worlds: enabling a high level
of automation during proof construction, while supporting formal proofs that can
capture rich incentive properties. To demonstrate our approach, the primary
technical contribution of our paper is a challenging case study: a formal proof
of Bayesian incentive compatibility (BIC) for the generic reduction from
algorithm design to mechanism design by \citet{HKM11}.  This example is an
attractive proof-of-concept for several reasons.
\begin{enumerate}
  \item Both the reduction and the proof of Bayesian incentive compatibility are
    complex. The mechanism is far from obviously strategy proof---indeed, the
    proof is a research contribution first published at SODA 2011.
  \item It is a general reduction, so certifying its correctness once certifies
    the incentive properties for any instantiation of the reduction.
  \item It relies on truthfulness of the Vickrey-Clarke-Groves (VCG) mechanism.
    As part of our efforts, we provide the first formal verification of
    truthfulness for this classical mechanism.
  \item It employs randomization both within the algorithm and within the agent
    behavior---agent types are drawn from the known Bayesian prior. 
\end{enumerate}
The formal proofs bear a resemblance to the original proof, both easing
formalization and making the proofs more accessible to the mechanism design
community. 

To formalize the proofs, we adapt techniques from program verification.  We view
incentive properties as a property of the mechanism and the agent's payoff
function, both expressed as programs. Formal verification has developed
sophisticated tools and techniques for verifying program properties, but
general-purpose tools require significant manual work. Verifying even moderately
complex mechanisms seems well beyond the reach of current technology. To ease
the task, we view incentive properties as \emph{relational properties}:
statements about the relationship between the outputs in two runs of the same
program. Specifically, consider the program which calculates an agent's payoff
under the mechanism and assume agents play their true value in the first run,
while an agent may deviate arbitrarily in the second run. If the output in the
first run is at least the output in the second run, then the mechanism is
incentive compatible.

With this point of view, we can use tools specialized for relational properties.
Such tools are significantly easier to use and have achieved notable successes
for verifying proofs from differential privacy and cryptography.  We use
\THESYSTEM, a recently-developed programming language that can
express and check relational properties \citep{BartheGAHRS15}.  \THESYSTEM has
been used to verify differential privacy and basic truthfulness in simple
mechanisms under complete information, like the fixed price auction and the
random sampling mechanism of \citet{GHKSW06} for digital goods.

Our work goes significantly beyond prior efforts in several respects. First, the
mechanism we verify is significantly more complex than previously considered
mechanisms, and we analyze all uses of the reduction, rather than just a single
instance. Second, the mechanism operates in the partial information setting, so
the proof requires careful reasoning about randomization (from both the
mechanism and from the prior distribution on types).

The main strength of our approach lies in the high degree of automation during
\emph{proof construction}. Once the mechanism and payoff functions have been
encoded as programs, and once we have supplied some annotations, we can
construct most of the formal proof automatically with the aid of automated
solvers. However, there are a handful of particularly complex steps that
\THESYSTEM fails to automatically prove. To finish the proof, we manually build
a formal proof for these missing pieces using EasyCrypt, a proof assistant for
relational properties, and Coq, a general purpose proof assistant.\footnote{%
  Our formal proofs, along with code for the \THESYSTEM tool, are available
  online: \url{https://github.com/ejgallego/HOARe2/tree/master/examples/bic}}

\paragraph*{Related work.}
%
Closely related to our work, a recent paper by \citet{CKLR15} uses the theorem
prover Isabelle to verify basic properties of the celebrated
Vickrey-Clarke-Groves (VCG) mechanism. They consider general auction properties:
the prices should be non-negative, VCG should produce a partition of goods, etc.
Moreover, their framework can be used to automatically produce a correct,
executable implementation of the mechanism.  While their work demonstrates that
formal verification can be applied to verify properties of mechanisms, their
results are limited in two respects.  First, they do not consider incentive
properties, arguably the properties at the heart of mechanism design. Second,
they apply general techniques from computer-aided verification that are not
specifically tailored to mechanism design, requiring substantial effort to
produce the machine-checked proof. Our work uses verification techniques that
are particularly suited for incentive properties.

In the
\iffull appendix \fi
we provide a primer on formal verification and discuss related work; a recent
survey by \citet{Kerber0R16} provides a comprehensive review of formal methods
for verifying mechanism design properties. The algorithmic game theory
literature has for the most part ignored the problem of \emph{verifying}
incentive properties, with a few notable exceptions. Recently, \citet{BP14}
define \emph{verifiably truthful mechanisms}.  Informally, such a mechanism is
selected from a fixed family of mechanisms such that for every truthful
mechanism in that family, a certificate showing truthfulness can be found in
polynomial time. \citet{BP14} consider mechanisms represented as decision trees
and show that for the one-dimensional facility location problem, truthfulness
for mechanisms in this class can be efficiently verified by linear programming.
In contrast, we investigate significantly more complex mechanisms in exchange
for forgoing worst-case polynomial time complexity.

\citet{Mua05} considers the problem of \emph{property testing} for truthfulness
in single parameter domains, which reduces to testing for a variant of
monotonicity. \citet{Mua05} gives a tester that can test whether there exist
payments that guarantee that truthful reporting is a dominant strategy with
probability $1-\epsilon$, given a poly$(1/\epsilon)$ number of arbitrary
evaluations of an allocation rule and assuming agents have uniformly random
valuations. In contrast, we assume direct access to the code specifying the
auction instead of merely black box access to the allocation rule, and we
achieve verification of exact truthfulness, not just approximate truthfulness.
We are also able to verify mechanisms in more complex settings, e.g., arbitrary
type spaces, randomized mechanisms, and arbitrary priors.

Our work is also related to the literature on automated mechanism design,
initiated by \citet{CS02} (see \citet{San03} or \citet[Chapter 6]{Con06} for an
introduction). In broad strokes, automated mechanism design seeks to generate
truthful mechanisms which optimize the designer's objectives. This is often
accomplished by solving explicitly for the distribution on outcomes defining a
mechanism using a mixed integer linear program encoding the incentive
constraints and objective, an NP hard problem that can often be solved
efficiently on typical instances \citep{Con06}.
Automated mechanism design targets a more difficult problem than we do: it seeks
not just to \emph{verify} the truthfulness of a given mechanism, but to
\emph{optimize} over all truthful mechanisms. However, these techniques have
some limitations: they produce explicit representations of mechanisms requiring
size exponential in the number of bidders, and they use an explicit integer
linear program, requiring a finite type space.
In contrast, by only requiring full automation for proof verification and not
proof construction, we are able to use a much more sophisticated
toolkit---including symbolic manipulation, not just numeric optimization---and
verify significantly more complex mechanisms that can have infinite outcome and
type spaces.

\section{Main example: RSM}
\label{sec:HKM}

As our main proof of concept, we verify that the Replica-Surrogate-Matching
(RSM) mechanism due to \citet{HKM11} is Bayesian incentive compatible. The RSM
mechanism reduces mechanism design to algorithm design: given an algorithm $A$
that takes in agents' reported types and selects an outcome, the RSM mechanism
turns $A$ into a Bayesian incentive compatible mechanism. Accordingly, our
formal proof will carry over to any instantiation of RSM.  We first review the
original proof by \citet{HKM11}. Then, we describe our verification process,
from pseudocode to a fully verified mechanism.


Let's begin with the standard notion of Bayesian incentive compatibility. We
assume there are $n$ agents, each with a \emph{type} $t_i$ drawn from some set
of types $T$.  Furthermore, we have access to a distribution $\mu$ on types, the
\emph{prior}.
A \emph{mechanism} is a (possibly randomized) function from the inputs---one per
agent---to a single \emph{outcome} $o$ from set $O$, and a real-valued
\emph{payment} $p_i$ for each agent.  Without loss of generality, we will assume
that the agents each report a type from $T$ as their input.  Agents have a
valuation $v(t,o)$ for type $t$ and outcome $o$. Agents have \emph{quasi-linear
  utility}: their utility for outcome $o$ and payment $p$ is $v(t, o) - p$.  We
will write $(s, t_{-i})$ for the vector  obtained by inserting $s$ into the
$i$th slot of $t$. Then, we want to check the following property.

\begin{definition}
  A mechanism $M$ is \emph{Bayesian incentive compatible (BIC)} if for every
  agent $i$ and types $t_i, t_i'$, we have
  \[
    \mathbb{E}_{t_{-i} \sim \mu^{n-1}} [ v(t_i, M(t_i, t_{-i})) - p_i(t_i, t_{-i}) ]
    \geq
    \mathbb{E}_{t_{-i} \sim \mu^{n-1}} [ v(t_i, M(t_i', t_{-i})) - p_i(t_i', t_{-i}) ] .
  \]
  The expectation is taken over the types $t_{-i}$ of the other agents (drawn
  independently from $\mu$) and any randomness used by the mechanism.
\end{definition}

\subsection{The RSM mechanism}

Now, let's consider the mechanism: the RSM mechanism in the ``idealized model''
by \citet{HKM11}. We will first recapitulate their proof, before explaining in
detail how we verify it.


RSM is a construction for turning an \emph{algorithm} $A : T^n \to O$ into a BIC
mechanism. The process is easy to describe: each agent individually transforms their
type $t_i$ to a \emph{surrogate type} $s_i$ by applying the
Replica-Surrogate-Matching procedure $R$.  This procedure also produces a
payment $p_i$ for the agent. Then, the surrogates $s$ are fed into the
algorithm $A$, which produces the final outcome.

\begin{figure}
  \begin{enumerate}
    \item Pick $i$ uniformly at random from $[m]$;
    \item Build a \emph{replica type profile} $\vec{r}$ by taking $m-1$ samples from
      $\mu$ for $\vec{r}_{-i}$, setting $r_i=t$;
    \item Build a \emph{surrogate type profile} $\vec{s}$ by taking $m$ samples from $\mu$;
    \item Build a bipartite graph with nodes $\vec{r}$, $\vec{s}$, and edges with
      weight
      \[
        w(r, s) = \mathbb{E}_{t_{-i} \sim \mu^{n-1}} [ v(r, A(s, t_{-i}))] ;
      \]
    \item Run the VCG procedure on the generated graph, and return the
      surrogate $s$ that is matched to the replica in slot $i$, and the
      appropriate payment $p$.
  \end{enumerate}

  \caption{Procedure $R$ with parameter $m$}
  \label{fig:rsm-algo}
\end{figure}

The procedure $R$ is described in \Cref{fig:rsm-algo}. Let $m$ be an integer
parameter---the number of replicas.
Given input type $t$, we take $m - 1$ independent samples from $\mu$, the
\emph{(r)eplicas}. We then take $m$ independent samples from $\mu$, the
\emph{(s)urrogates}. Finally, we select an index $i$ uniformly at random from
$[m]$, and place the original type $t$ in the $i$th ``slot'' of the replicas
$\vec{r}$.
%
We will consider the replicas as ``buyers'', and the surrogates as ``goods'',
and assign a numeric ``value'' for every pair of buyer and good. The value of
replica $r$ for surrogate $s$ is set to be
\begin{equation}
  \label{eq:weight}
  w(r, s) = \mathbb{E}_{t_{-i} \sim \mu^{n-1}} [ v(r, A(s, t_{-i}))] ,
\end{equation}
that is, the expected utility of an agent with true type $r$ reporting type $s$.
Finally, RSM runs the well-known Vickrey-Clarke-Groves mechanism
\citep{vickrey61,clarke71,groves73} to match each replica with a surrogate in
this market. The output is the surrogate matched to replica in slot $i$ (the
original type $t$), along with the payment charged.
\paragraph*{The original proof.}

The proof of BIC from \citet{HKM11} proceeds in two steps.  First, they show
that $R$ is \emph{distribution preserving}.

\begin{lemma}[\citet{HKM11}] \label{lem:distr-pres}
  Sampling a type $t \sim \mu$ as input to $R$ gives the same distribution
  ($\mu$) on the surrogates output.
\end{lemma}

\begin{proof}
  When $R$ constructs the list of buyers before applying VCG, the distribution
  over buyers is simply $m$ independent samples from $\mu$, no matter the value
  of $i$. So, we can delay sampling $i$ and selecting the surrogate until after
  running VCG (via the principle of deferred decision).
  VCG produces a perfect matching of replicas to surrogates, and the surrogates
  are also $m$ independent samples from $\mu$. So, sampling a random replica $i$
  and returning the matched surrogate is an unbiased sample from $\mu$.
\end{proof}

With the lemma in hand, \citet{HKM11} show that RSM is BIC.

\begin{theorem}[\citet{HKM11}] \label{thm:bic}
  The RSM mechanism is BIC.
\end{theorem}

\begin{proof}
  Consider bidder $i$ with type $t_i$, and fix the randomness for bidder $i$. In
  the VCG procedure of $R$, the value of $i$'s replica for surrogate $s$ is
  $w(t_i, s)$: the expected utility for submitting $s$ to $A$ while having true
  type $t_i$, assuming that all other inputs to $A$ are drawn from $\mu$.

  In the RSM mechanism, the other inputs to $A$ are computed by sampling a type
  $t_j \sim \mu$, and taking the surrogate produced by $R(t_j)$. By
  Lemma \ref{lem:distr-pres}, the distribution over surrogates is $\mu$. Therefore,
  $w(t_i,s)$ is bidder $i$'s expected utility in the RSM mechanism for ending up
  matched to $s$. Since VCG is incentive compatible, bidder $i$ has no incentive
  to deviate to any other bid $t_i'$. By taking expectation over the randomness
  of $i$, we get the result.
\end{proof}

Crucially, \Cref{thm:bic} relies on the truthfulness property of the VCG
mechanism.  We have also verified this property but we postpone our discussion
to 
\iffull \Cref{sec:vcg}\fi; the verification of RSM is
more interesting.

\section{Verifying RSM}
\label{sec:verif}
Now that we have seen the mechanism, we present our verification step by step.

\begin{minipage}[l]{0.6\textwidth}
\hspace{-1cm}
\begin{enumerate}
  \item We write the RSM mechanism as a program in the \THESYSTEM programming
    language.
  \item We annotate the program with assertions expressing the BIC
    property, and some additional intermediate facts (lemmas).
  \item The tool automatically generates the \emph{verification
      conditions} (VCs), which imply BIC.
  \item The tool uses automatic solvers to check the verification conditions;
    they may fail to prove some assertions.
  \item Finally, we prove the remaining verification conditions by
    using an interactive prover.
\end{enumerate}
\end{minipage}
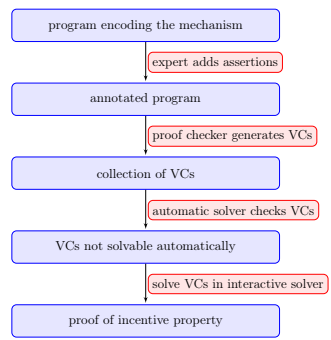
\begin{minipage}[c]{0.4\textwidth}
\hspace{.5cm}
\scalebox{.45}{  \begin{tikzpicture}[auto,
    pdata/.style={rectangle, draw=blue, thick, fill=blue!10,
      text width=20em, text centered, rounded corners, minimum height=2.5em},
    pdatam/.style={rectangle, draw=blue, thick, fill=blue!20,
      text width=6em, text centered, rounded corners, minimum height=2.5em},
    line/.style={draw, thick, -latex',shorten >=1pt},
    relline/.style={draw, thin, -latex', dashed, double},
    ann/.style={rectangle, draw=green, thick, rectangle, rounded
      corners, fill=green!20, outer sep=0.34em, solid},
    pcomponent/.style={midway, draw=red, thick, rectangle, rounded corners,
      fill=red!10, minimum height=1.5em, outer sep=0.2em}]

    \matrix [column sep=35mm, row sep=12mm, every node/.style=pdata]
    {
      \node                (program)     {program encoding the mechanism};\\
      \node                (annotated)   {annotated program};\\
     \node                (verification)  {collection of VCs};\\
      \node             (automatic)       {VCs not solvable automatically}; \\
  \node                (interactive)       {proof of incentive property};\\
    };
    \begin{scope}[every path/.style=line,
                  every node/.style=pcomponent]
      \path          (program)    -- node         {expert adds assertions} (annotated);
      \path          (annotated) -- node  {proof checker generates
        VCs} (verification);
      \path          (verification)      -- node
      {automatic solver checks VCs} (automatic);
      \path (automatic)  -- node   {solve VCs in interactive solver} (interactive);
    \end{scope}
  \end{tikzpicture}
}
\end{minipage}

\medskip

The outcome of these five steps is a formal proof that the RSM
mechanism enjoys the BIC property. In the following, we will combine the
description of different steps in the same subsection.

\subsection*{Step 1: Modeling the mechanism}

To express RSM as a program, we will code a single agent's utility function when
running the RSM mechanism, when all the other agents report truthfully and have
types drawn from $\mu$.  Remembering that we consider truthfulness as
a \emph{relational property}, we will then reason about what happens when the agent
reports truthfully, compared to what happens when the agent deviates.

We model types and outcomes as drawn from (unspecified) sets $T$ and $O$, and we
assume an algorithm \verb|alg| mapping $T^n \to O$. We will consider what
happens when the first bidder deviates. This is without loss of generality: if
$j$ deviates, we can consider the RSM mechanism with \verb|alg| replaced by a
version \verb|alg'| that first rotates the $j$th bidder to the first slot, when
proving BIC for the first bidder under \verb|alg'| implies BIC for the $j$th
bidder under $A$.
For the values, we will assume an arbitrary valuation function \verb|value|
mapping $T \times O \to \mathbb{R}$. In the code, we will write \verb|mu| for
the prior distribution $\mu$.



Let's begin by coding the RSM transformation $R$, which transforms an agent's
type into a surrogate type and a payment. It will be convenient to separate the
randomness from $R$. We encode $R$ as a deterministic function \verb|Rsmdet|,
which takes as input the agent number \verb|j|, the random coins \verb|coins|,
and the input type \verb|report|.  We will have \verb|Rsmdet| take an additional
parameter \verb|truety| representing an agent's true type. This variable does
not show up in the code--RSM does not have access to this information---but will
be useful later for expressing Bayesian incentive compatibility as a relational
property.  We will model the slot as a natural number.

In 
\iffull \Cref{sec:vcg} \fi
we will discuss our treatment of VCG in more detail, but it is enough to know
that VCG takes a list of buyers and a list of goods.  VCG will output a
permutation of goods (representing the assignment), and a corresponding list of
payments.
\begin{figure}
  \begin{minipage}[b]{0.49\textwidth}
\begin{lstlisting}[style=numbers]
def Rsmdet(j, coins, truety, report) =
  (rs$_{-i}$, ss, i) = coins;
  vcgbuyers = (report, rs$_{-i}$);
  (surrs, pays) = Vcg(vcgbuyers, ss);
  return (surrs[j], pays[j])
\end{lstlisting}
\caption{Defining RSM}
\label{fig:rsm}
\end{minipage}
\begin{minipage}[b]{0.49\textwidth}
\begin{lstlisting}[style=numbers]
def Expwts(j, r, s) =
  sample others$_{-j}$ = mu$^{n - 1}$;
  algInput = (s, others$_{-j}$);
  outcome = alg(algInput);
  return expect_num { value(r, outcome) }
\end{lstlisting}
\caption{Defining weights}
\label{fig:weights}
\end{minipage}
\end{figure}
In \Cref{fig:rsm}, bolded words are keywords and primitive operations of
\THESYSTEM.  For a brief explanation, line (2) names the three components of
\verb|coins|: the replicas \verb|rs|$_{-i}$, the surrogates \verb|ss|, and the
slot \verb|i|; line (3) puts the agent's input type \verb|report| in the proper
slot for the replicas; line (4) calls VCG on the list of buyers \verb|vcgbuyers|
produced at line (3) and the list of surrogates \verb|ss| as goods; and line (5)
returns the surrogate and payment.  

The \verb|Expwts| function in \Cref{fig:weights} implements  the $w$
function from \Cref{eq:weight}, with the additional parameter \verb|j| to
indicate the agent. In \Cref{fig:weights}, line (2) samples $n-1$ types
\verb|others|$_{-j}$ from $\mu$ for the other agents. These are the types on
which the expectation is taken in \Cref{eq:weight}. Line (4) uses the algorithm
\verb|alg| to compute the outcome \verb|outcome| when the agent \verb|j| report
type \verb|s|. Finally, the \verb|expect_num| on line (5) takes the expected
value of the distribution over reals defined by evaluating the value function
\verb|value| on the true type \verb|r| and on the randomized outcome of the
\verb|alg|.

To check the BIC property, we will code the expected utility for the first
bidder and then check that it is maximized by truthful reporting. To
break down the code, we will suppose that the function takes in a list of
functions \verb|othermoves| that transform each of the other bidder's type.

\begin{figure}[h!]
  \begin{minipage}[b]{0.49\textwidth}
\begin{lstlisting}[style=numbers]
def Util(othermoves, myty, mybid) =
  return (expect rsmcoins Helper)

def Helper(coins) =
  (mysur, mypay) =
    Rsmdet(1, coins, myty, mybid);
  myval = expect_num {
    for $i = 1 \dots n-1$:
      sample othersurs[i] =
        (sample otherty = mu;
        othermoves[i](otherty));
    algInput = (mysur, othersurs);
    outcome = alg(algInput);
    value(myty, outcome) };
  return (myval - mypay)
\end{lstlisting}
\caption{Defining utility}
\label{fig:utility}
\end{minipage}
\begin{minipage}[b]{0.49\textwidth}
\begin{lstlisting}[style=numbers]
def Others(j, t) =
  sample coins = rsmcoins;
  (s, p) = Rsmdet(j, coins, t);
  return s

def MyUtil(ty,bid) = Util(Others,ty,bid)
\end{lstlisting}
\caption{Defining other reports}
\label{fig:others}
\end{minipage}
\end{figure}

The distribution \verb|rsmcoins| defines the distribution over the coins to $R$,
i.e., sampling the replicas $\vec{r}_{-i}$, the surrogates $\vec{s}$, and the
coin $i$.  We encoded this distribution in \THESYSTEM, but we elide it for lack
of space.  In the code in \Cref{fig:utility}, on line (2) we take expectation of
the function \verb|Helper| over the distribution \verb|rsmcoins|, with
\verb|expect|.  In \verb|Helper|, we then call \verb|Rsmdet| on line (6) to
compute the surrogate and payment for the agent, passing $1$ since we are
calculating the utility for the first agent.  We sample the other agents' types
and transform them on lines (9--11), and we take expectation of the first
agent's value for the outcome on lines (7--14).  Finally, we subtract off the
payment on line (15), giving the final utility for the first agent.

To complete our modeling of RSM, in \Cref{fig:others} we plug in \verb|Others|
into the utility function: it simply takes an agent number and a type as input,
samples the coins from \verb|rsmcoins|, and returns the surrogate from calling
\verb|Rsmdet|.
So far, we have just written code describing how to implement the RSM mechanism
and how to calculate the utility for a single bidder. Now, we express the BIC
property as a property about this program and check it with \THESYSTEM.

\subsection*{Step 2: Adding assertions}
We specify properties in \THESYSTEM by annotating variable and functions with
assertions of the form $\{ x :: Q \mid \phi \}$, read as ``$x$ is an element of
set $Q$ and satisfies the logical formula $\phi$''. These assertions serve two
purposes: (1) they express facts to be proved about the code and (2) they assert
mathematical facts about primitive operations like \verb|expect| and
\verb|expect_num|. The system will then formally verify that the first kind of
annotations are correct, while assuming the assertions of the second kind as
axioms.

A key feature of \THESYSTEM is that the assertion $\phi$ is \emph{relational}:
it can refer to two copies of each variable $x$, denoted $\l{x}$ and $\r{x}$.
Roughly, we may make assertions about two runs of the same program where in the
first program we use variables $\l{x}$, and in the second run we use variables
$\r{x}$.\iffull\footnote{%
  These annotations are known as \emph{relational refinement types} in the
  programming language literature. We will call them assertions or annotations
  to avoid clashing with agent types.}\fi{}
For instance, truthfulness corresponds to the following assertions:
\begin{align}
  \{ ty :: T \mid \l{ty} = \r{ty} \} &
  \tag{true type is equal on both runs} \\
  \{ bid :: T \mid \l{bid} = \l{ty} \} &
  \tag{bid is the true type in the first run} \\
  \{ utility :: \mathbb{R} \mid \l{utility} \geq \r{utility} \} &
  \tag{utility is higher in the first run}
\end{align}

Our goal is to check these assertions for the function \verb|MyUtil|, which
computes an agent's utility in expectation over the other types. Along the way
we will use several intermediate facts, encoded as assertions in \THESYSTEM.
Assertions on primitive operations, like \verb|expect| and \verb|expect_num|,
are the axioms. Assertions on larger chunks of code are proved correct from the
assertions on the subcomponents.

\paragraph*{Monotonicity of expectation.}

Since the BIC property refers to \emph{expected} utility, we use an expectation
operation \verb|expect| when computing an agent's utility (line (2) of the
\verb|Util| code).  To show BIC, we need a standard fact about
\emph{monotonicity} of expected value: for functions $f \leq g$, $\mathbb{E}[f]
\leq \mathbb{E}[g]$ taken over the same distribution. This can be encoded with
an annotation for \verb|expect|:
\begin{align*}
  \Distr{\{ c :: C \mid \l{c} = \r{c} \}}
\rightarrow
\{ f :: C \to \mathbb{R} \mid \forall x.\ \l{f}(x) \leq \r{f}(x) \}
  \rightarrow
  \{ e :: \mathbb{R} \mid \l{e} \leq \r{e} \} .
\end{align*}

This annotation indicates that \verb|expect| is a function that takes two
arguments and returns a real number. In each of the three components, the
annotation before the bar specifies the type of the value: The first argument
is a distribution over $C$, the second argument is a real-valued function $C \to
\mathbb{R}$, and the return value is a real number. The logical formulas after
the pipe describe how two runs of the expectation function are related. The
first component states that in the two runs, the distributions are the same.
The second component states that the function $f$ in the first run is pointwise
less than $f$ in the second run. The final component asserts that the expected
value---a real number---is less on the first run than on the second run.

If think of the distribution as being over the coins \verb|rsmcoins|, this fact
allows us to prove deterministic truthfulness for each setting of the coins,
then take expectation over the coins in order to show truthfulness in
expectation. This is what we need to prove for the BIC property, and is
precisely the first step in the original proof of \Cref{thm:bic}.

\paragraph*{Distribution preservation.}

When we consider a single agent, truthful bidding may not be BIC for arbitrary
transformations of the other agents' types (\verb|othermoves| in the \verb|Util|
code). As indicated by \Cref{lem:distr-pres}, we also need the transformation to
be distribution preserving: the output distribution on surrogates must be the
same as the distribution on input types.

Much as we did above, we can capture this property with appropriate annotations.
While we have so far used rather simple formulas $\phi$  that only mention
variables in $\{ x :: T \mid \phi \}$, in general the formulas $\phi$ can
describe assertions about programs.\iffull\footnote{%
  Of course, we need to actually \emph{check} the assertions eventually, whether
  by automated solvers or manual techniques.}\fi{}
We can annotate the \verb|othermoves| argument to \verb|Util| to require
distribution independence:
\[
  \{
  \lstt{othermoves} : \List{(T \to \Distr{T})}
  \mid
  \forall j \in [n].\; (\lstt{sample ot = mu; othermoves[j](ot)}) = \lstt{mu}
  \}
\]
To read this, \verb|othermoves| is a list of functions $f_j$ that take a type and
return a distribution on types, such that if we sample a type from \verb|mu|
and feed it to $f_j$, the resulting distribution (including randomness over the
initial choice of type) is equal to \verb|mu|. In other words, this asserts the
distribution preservation property of \Cref{lem:distr-pres} for each of the
other agent's transformations.

\paragraph*{Facts about VCG.}
Recall that \verb|Vcg| takes a list of bidders and a list of goods, and produces
a permutation of the goods and a list of payments as output. In our case, the
bidders and goods are both represented as types in $T$, so we can annotate the
\verb|Vcg| as:
\[
  \{ buys :: \List{T} \} \to
  \{ goods :: \List{T} \} \to
  \{ (alloc, pays) :: \List{T} \times \List{\mathbb{R}} \mid \lstt{vcgTruth} \land
  \lstt{vcgPerm}
  \} .
\]
The two assertions  $\lstt{vcgTruth}$ and  $\lstt{vcgPerm}$ reflect two facts
about VCG. The first is that VCG is incentive compatible; this can be encoded
like we have already seen, with a slight twist: We require that VCG is IC for a
deviation by \emph{any} player rather than just the first player, since
the possibly deviating player may be in any slot.
More precisely, we define the formula
\begin{align*}
  \lstt{vcgTruth} &:= \forall \lstt{j} \in [\lstt{m}].\;
  (\lstt{bids}_{-j, 1} = \lstt{bids}_{-j, 2}) \Longrightarrow
  \\
  & \lstt{Expwts(j, bids$_1$[j], alloc$_1$[j])} - \lstt{pays$_1$[j]}
  \geq
  \lstt{Expwts(j, bids$_1$[j], alloc$_2$[j])}) - \lstt{pays$_2$[j]} .
\end{align*}
We treat the bid in the first run (\lstt{bids$_1$[j]}) as the true type, and the
bid on the second run (\lstt{bids$_2$[j]}) as a possible deviation---this is why
we evaluate the $j$th bidder's expected utility using the same true type in the
two runs.
The second fact we use is that VCG \emph{matches} buyers to the goods. In fact,
since the number of goods (surrogates) and the number of buyers (replicas) are
equal, VCG produces a perfect matching. We express this by asserting that VCG
outputs an assignment that is a permutation of the goods:
\[
  \lstt{vcgPerm} :=
  \lstt{isPerm goods$_1$  alloc$_1$}
  \land
  \lstt{isPerm goods$_2$  alloc$_2$} .
\]
We verify these properties for a general version of VCG. The verification
follows much like the current verification; we discuss the details in
\iffull \Cref{sec:vcg}\fi.

\subsection*{Step 3: Handling proof obligations}
After providing the annotations, \THESYSTEM is able to automatically check most
of the annotations with \emph{SMT solvers}\iffull\footnote{%
  Satisfiability-Modulo-Theory, see e.g., \citep{handbook-sat} for a survey.}\fi---fully
automated solvers that check the validity of logical formulas.  Such solvers are
a staple of modern formal verification. While the underlying problem is often
undecidable, modern solvers employ sophisticated heuristics that can efficiently
handle large formulas in practice.


We are able to use SMT solvers to automatically check all but three proof
obligations; for these three facts the SMT solvers time out without finding a
proof. The first two are uninteresting, and we manually construct the formal
proof using the Coq proof assistant. The last obligation is more interesting: it
corresponds to \Cref{lem:distr-pres}.  Concretely, when we define an agent's
expected utility
\[
\lstt{def MyUtil(ty,bid) = Util(Others,ty,bid)} ,
\]
recall that \verb|Util| asserts that \verb|Others| is distribution
preserving. This is precisely \Cref{lem:distr-pres},
and the automated solvers fail to prove this automatically.

To handle this assertion we use a more manual tool called EasyCrypt
\citep{Barthe11,BartheDGKSS13}, a proof assistant that allows the user prove
equivalence of two programs $A$ and $B$ by manually transforming the source code
of $A$ until the source code is identical to $B$.\iffull\footnote{%
  This is a common proof technique in cryptographic proofs, known as \emph{game
    hopping} \citep{Bellare:2006,Halevi:2005}.}\fi{}
We prove that \verb|Others| is equivalent to the program that simply samples
from \verb|mu| by transforming the code for \verb|Others| (including the code
sampling the coins of the mechanism, \verb|rsmcoins|) in several stages. We
present the code in \Cref{fig:games} with two replicas, for simplicity.

\begin{figure}
\begin{minipage}[t]{0.21\textwidth}
\begin{lstlisting}[style=nonumbers]
def stage1 =
 sample ot = mu;
 Others(ot)
\end{lstlisting}
\end{minipage}
\begin{minipage}[t]{0.27\textwidth}
\begin{lstlisting}[style=nonumbers]
def stage2 =
 sample ot = mu;
 sample r' = mu;
 sample s1 = mu;
 sample s2 = mu;
 sample i = flip;

 if i then
  (r1,r2) = (ot,r');
 else
  (r1,r2) = (r',ot);

 bs = (r1,r2);
 gs = (s1,s2);

 (ss,ps) = Vcg(bs,gs);
 (o1,o2) = ss;

 if i then o1 else o2
\end{lstlisting}
\end{minipage}
\begin{minipage}[t]{0.27\textwidth}
\begin{lstlisting}[style=nonumbers]
def stage3 =
 sample ot = mu;
 sample r' = mu;
 sample s1 = mu;
 sample s2 = mu;

 (r1,r2) = (ot,r');

 bs = (r1,r2);
 gs = (s1,s2);

 (ss,ps) = Vcg(bs,gs);
 (o1,o2) = ss;

 sample i = flip;
 if i then o1 else o2
\end{lstlisting}
\end{minipage}
\begin{minipage}[t]{0.22\textwidth}
\begin{lstlisting}[style=nonumbers]
def stage4 =
 sample s1 = mu;
 sample s2 = mu;
 sample i = flip;
 if i then s1
      else s2
\end{lstlisting}
\end{minipage}
\caption{Code transformations to prove \Cref{lem:distr-pres}.
  \label{fig:games}}
\end{figure}

The proof boils down to showing that each step transforms a program to an
equivalent program. Our starting point is \verb|stage1|, the program that
samples an agent's type from \verb|mu| and runs \verb|Others| on the sampled
value.  Unfolding the definition of \verb|Others|, \verb|Rsmdet|,
\verb|rsmcoins| and including the code that puts the agent's input type in the
proper slot for the replicas, we obtain program \verb|stage2|. From there, the
main step is to show that we don't need to place the replicas in a random order
before calling \verb|Vcg|.  Then, we can move the sampling for \verb|i| down
past the \verb|Vcg| call, giving \verb|stage3|.  Finally, using the fact that
the output assignment \verb|ss| of \verb|Vcg| is a permutation of the goods
\verb|(s1, s2)|, we obtain the program  \verb|stage4| and conclude that this is
equivalent to taking a single sample from \verb|mu|. This chain of
transformations has been verified with EasyCrypt.

\section{Perspective}
\label{sec:perspective}

Now that we have presented our verification of the RSM mechanism, what have we
learned and what does formal verification have to offer mechanism design going
forward?  In our experience, while formal verification of game theoretic
mechanisms is by no means trivial, verification tools are maturing to a point
where practical verification of complex mechanisms can be envisioned.
Our verification of RSM, for instance, involved only coding the utility function
and adding annotations, most of which can be checked automatically. The most
time-consuming part was manually proving the last few assertions.

At the same time, the range of mechanisms that can be verified is less clear.
There is an art to encoding a mechanism in the right way, and some mechanisms
are easier to verify than others. Since we are trying to verify proofs, the
crucial factor is the complexity of the \emph{proof} rather than the complexity
of the mechanism. Clean proofs where, each step reasons about localized parts of
the program, are more amenable to verification; proof patterns---like universal
truthfulness---also help.

In sum, formal verification can manage the increasing complexity of mechanisms
by formally proving incentive properties for everyone---mechanism designers,
mechanism users, and even mechanism programmers. We believe that the tools to
verify one-shot mechanisms are already here. So, we propose a challenge: Try
using tools like \THESYSTEM to verify your own mechanisms, putting formal
verification techniques to the test. We hope that one day soon, verification for
mechanisms will be both easy and commonplace.

\paragraph*{Acknowledgments.}
We thank the anonymous reviewers for their careful reading; their suggestions
have significantly improved this work. We especially thank Ran Shorrer for
pointing out empirical evidence that agents may manipulate their reports even
when the mechanism is truthful. This work was partially supported by NSF grants
TWC-1513694, CNS-1237235, CNS-1565365  and a grant from the Simons Foundation
($\#360368$ to Justin Hsu).




\putbib[header,refs]
\end{bibunit}

\iffull
\newpage

\appendix

\begin{bibunit}[abbrvnat]

\section{A note about worst-case complexity}
As is typical in program verification, we distinguish between constructing a
proof and checking it. Constructing the proof is hard: we do not assume that a
proof (or some representation, like a certificate) can be found automatically in
worst-case polynomial time, and we will even allow a human to play a limited
part in this process. However, checking the proof must be easy: agents should be
able to take the formal proof and check it in polynomial time.


While worst-case polynomial time for the entire process would be ideal, it is
not very realistic as we cannot expect an algorithm to prove the incentive
properties automatically---the proof may be a research contribution; deciding
whether an incentive property holds at all may be an undecidable problem.
However, relaxing the running time condition when constructing the proof is
well-motivated in our application. Unlike the mechanism itself, the proof
construction procedure will not be run many times on inputs of unknown origin
and varying size. Instead, for a particular mechanism, the proof is constructed
just once. In exchange for relaxing worst-case running time, we can verify rich
classes of mechanisms.

\section{Verifying the VCG Mechanism}
\label{sec:vcg}

The celebrated VCG mechanism is a foundational result in the mechanism design
literature. It calculates an outcome maximizing social welfare (i.e., the sum of
all the agents' valuations) and payments ensuring that truthful bidding is
incentive compatible. Let's briefly review the definition.

\begin{definition}[\citet{vickrey61,clarke71,groves73}]
  Let $O$ be a space of outcomes, and let $v : T \times O \to \mathbb{R}$ map
  agent types and outcomes to real values. Given a reported type profile $t$
  from $n$ agents, the VCG mechanism selects the \emph{social-welfare}
  maximizing outcome:
  \[
    o^* := \argmax_{o \in O} \sum_{i \in [n]} v(t_i, o) ,
  \]
  and computes payments
  \[
    p_j :=
    \max_{o \in O} \sum_{i \in [n] \setminus \{ j \}} v(t_i, o)
    -
    \sum_{i \in [n] \setminus \{ j \} } v(t_i, o^*) .
  \]
  That is, the payment for agent $j$ is the difference between the welfare for
  the other agents without $j$ present, and the welfare for the other agents
  with $j$ present.
\end{definition}

As Vickrey, Clarke, and Groves showed, this mechanism is incentive compatible.
Let's consider how to verify this fact in \THESYSTEM.  Like for RSM, we will
start by coding the utility function for a single bidder.  We will call it
\verb|VcgM| to distinguish it from the more special case we need for RSM;
\Cref{fig:vcg} presents the full code.
\begin{figure}[h!]
\begin{lstlisting}[style=numbers]
def VcgM(values, range) =
  welfare = sumFuns(values);
  outcome = findMax(welfare, range);

  for i = $1 \dots n$:
    welfWithout = sumFuns(values$_{-i}$);
    outWithout = findMax(welfWithout, range);
    prices[i] = welfWithout(outWithout) - welfWithout(outcome)
  end

  (outcome, prices)
\end{lstlisting}
\caption{Encoding the VCG mechanism in \THESYSTEM
  \label{fig:vcg}}
\end{figure}
The parameters to \verb|VcgM| are a list of valuation functions (\verb|values|)
and a set of possible outcomes (\verb|range|). We use two helper functions:
\verb|sumFuns| takes a list of valuation functions and sums them to form the
social welfare function, while \verb|findMax| takes a objective function and a
set of outcomes, and returns the outcome maximizing the objective.

To encode the incentive property, we will consider two runs of \verb|VcgM|. We
allow any single agent to deviate on the two runs. For the deviating agent, we
will model her report in the first run will be her``true'' valuation.  Then, we
want to give \verb|VcgM| the following annotation:
\[
  \{ \lstt{values} : O \to \mathbb{R} \}
  \to
  \{ \lstt{range} : \List{O} \}
  \to
  \{ (\lstt{out}, \lstt{pays}) : O \times \List{\mathbb{R}} \mid \lstt{out} \in \lstt{range} \land
  \lstt{vcgTruth} \} .
\]
The predicate \lstt{vcgTruth} captures truthfulness, like the assertion in
\Cref{sec:verif}:
\begin{align*}
  \lstt{vcgTruth} &:= \forall \lstt{j} \in [\lstt{m}].\;
  (\lstt{values}_{-j, 1} = \lstt{values}_{-j, 2}) \Longrightarrow
  \\
  & \lstt{values[j]$_1$(out$_1$[j])} - \lstt{pays$_1$[j]}
  \geq
  \lstt{values[j]$_1$(out$_2$[j])} - \lstt{pays$_2$[j]} .
\end{align*}
With appropriate annotations on \lstt{findMax}, \lstt{sumFuns}, and the
``all-but-$j$'' operation $(-)_{-j}$, \THESYSTEM verifies VCG automatically.

\section{A primer on program verification}
Program correctness and program verification have a venerable history.
In a visionary article,~\citet{Turing:1949} presents a rigorous proof
of correctness for a computer routine; although very short, this note
prefigures the current trends in deductive program verification and
introduces many fundamental ideas and concepts that still remain at
the core of program verification today. In particular, Turing makes a
clear distinction between the programmer and the verifier, and argues
that in order to alleviate the task of the verifier, the programmer
should annotate his code with \emph{assertions}, i.e.\, predicates on program
states. Moreover, Turing argues that it should be possible to verify
assertions locally and that the correctness of the routine should be
expressed by the initial and final assertions, i.e.\, the assertions
attached to the entry and exit points, which respectively capture
\emph{hypotheses} on the program inputs and \emph{claims} about the program outputs.

Leveraging contemporary developments in programming language theory,
the seminal works of \citet{Floyd67} and \citet{Hoare69} formalize
verification methods that adhere to the program proposed by Turing. Both
formalisms share similar principles and make a central use of
invariants for reasoning about programs with complex control-flow; for instance,
both methods use \emph{loop invariants}---assertions that
hold when the program enters a loop and remain valid during loop
iterations. However, the methods differ in the specifics of proving
program correctness. On the one hand, Hoare logic provides a \emph{proof
system}---a set of axioms and rules for combining axioms---that can be used
to build valid formal proofs that establish
program correctness. On the other hand, Floyd calculus computes---from
an annotated program---a set of \emph{verification conditions}: formulas of some
formal language such as first-order logic,
whose validity implies correctness of the program. Despite their differences,
the two approaches can be proved
equivalent, and assuming that the underlying language of assertions is
sufficiently expressive, are \emph{relatively complete}~\cite{Cook:1978};
relative completeness reduces the validity of program specifications
to the validity of assertions.

Both \citet{Floyd67} and \citet{Hoare69} are designed to reason about
\emph{properties}, i.e.\, sets of program executions. They cannot reason
about the larger class of \emph{hyperproperties}~\cite{Clarkson:2008},
which characterize sets of sets of program executions. Continuity
(small variations on the input induce small variations on the output),
and truthfulness (pay-off is maximized when agents play their true
value) are prominent \emph{binary} instances of hyperproperties, also
called \emph{relational properties}. Reasoning about relational
properties is challenging and the subject of active research in
programming languages. One way for reasoning
about such properties is by using relational variants
of \citet{Floyd67} and \citet{Hoare69}. These
variants~\cite{Benton:2004} reason about two programs (or two copies
of the same program) and use so-called \emph{relational
assertions}, assertions which describe pairs of states.

Another challenge in program verification is to deal with probabilistic
programs. Starting from the seminal work of \citet{Kozen:1985}, numerous logics
have been proposed to reason about properties of probabilistic programs,
including~\citep{Morgan:1996,Chadha:2007}.  Recently, \citet{Barthe:2009}
propose a relational logic for reasoning about probabilistic programs.
\citet{BartheGAHRS15} extend and generalize the relational logic to the setting
of a higher-order programming language.

In recent years, the theoretical advances in program verification have been
matched by the emergence of computer-aided verification tools that
have successfully validated large software developments.
Most tools implement algorithms for computing verification conditions;
the algorithms are similar in spirit to~\citet{Floyd67}, although they
typically use optimizations~\citep{flanagan2001}. Moreover, most systems
use fully automated tools to check that verification conditions are
valid. However, there is a growing trend to complement this process
with an interactive phase, where the programmer interactively builds a
proof of difficult verification conditions that cannot be
proved automatically. Contrary to automated tools, which try
to find a proof of validity, interactive tools try to
check that the proof of validity built interactively by the programmer
is indeed a valid proof. This interactive phase is often required for
proving rich properties of complex programs. It is also often
helpful for proving relational properties of probabilistic
programs~\cite{BartheDGKSS13}.





So far, our account of formal verification has focused on so-called
deductive methods: methods where the verification corresponds to build
formal proofs that can be constructed using a finite set of rules
starting from a given set of axioms. However, there
are many alternative
methods for proving program correctness. In particular, algorithmic
methods, such as model-checking, have been highly successful for
analyzing properties of large systems. Algorithmic methods are
fundamentally limited by the state explosion problem, since the
methods become intractable when the state space becomes too large.
Modern tools based on algorithmic verification exploit a number of
insights for alleviating the state explosion problem, including
symbolic representations of the state space, partial order reduction
techniques, and abstraction/refinement techniques.
%


\section{Related Work in Computer-aided Program Verification}

There is a small amount of work in the programming languages and computer-aided
program verification literature on verification of truthfulness in mechanism
design. \citet{BUCS-TR-2008-026} give an interesting approach, by presenting a
programming language for automatically verifying simple auction mechanisms.  A
key component of the language is a type analysis to determine if an algorithm is
\emph{monotone}; if bidders have a single real number as their value
(\emph{single-parameter domains}), then truthfulness is equivalent to a
monotonicity property (e.g., see \citet{mu2008truthful}). Their language can be
extended by means of user-defined primitives that preserve monotonicity. The
paper shows the use of the language for verifying two simple auction examples,
but it is unclear how this approach scales to larger auctions, and does not
extend beyond single parameter domains.

\citet{DBLP:conf/aaai/WooldridgeADH07} promote the use of automatic verification
techniques where mechanism design properties are described by means of
\emph{specification logics} (like Alternating Temporal Logic
\citep{Alur:2002:ATT}), and where the verification is performed in an automatic
way by using the \emph{model checking} technique.  Similarly,
\citet{DBLP:conf/ceemas/TadjouddineG07}  propose a similar approach where first
order logic is used as a specification logic. This approach works well for
simple auctions with few numbers of bidders but suffers from a state explosion
problem when the auctions are complex or the number of bidders is large. This
situation can be alleviated by combining different engineering techniques
\citep{TGV09}, but it is unclear if this approach can be scaled to handle
complex auctions with a large number of bidders. Moreover, these automatic
approaches do not work in setting of incomplete information.

An alternative approach based on \emph{interactive theorem proving} has been
explored by \citet{DBLP:conf/facs2/BaiTPG13}. Interactive theorem provers allows
specifying and formally reasoning about arbitrary auctions and different
truthfulness properties. More generally, they have been used to formalize large
theorems in mathematics \citep{DBLP:conf/itp/GonthierAABCGRMOBPRSTT13}.
Unfortunately, verifying the required properties can require advanced proof
engineering skills, even for very simple auctions;
\citet{DBLP:conf/facs2/BaiTPG13} consider the English and Vickrey auctions.

\iffull \else \scriptsize \fi
\putbib[header,refs]
\end{bibunit}
\fi

\end{document}

%% file: prelude.tex
\newcommand{\THESYSTEM}{\textsf{HOARe}$^2$\xspace}
\usepackage{xspace}
\usepackage{stmaryrd}
\usepackage{color}
\usepackage{float}
\usepackage{listings}
\usepackage{xfrac}
\usepackage[usenames,dvipsnames]{xcolor}
\usepackage{bbm}
\usepackage{url}
\usepackage{paralist}
\usepackage{bussproofs}
\usepackage{xifthen}
\usepackage{proof}
\usepackage{caption}

\DeclareCaptionFormat{myformat}{#1#2#3\hrulefill}
\usepackage{amsmath,amsfonts,amssymb}
\usepackage{mathabx}

\iffull
\usepackage{amsthm}

\theoremstyle{plain}
 \newtheorem{lemma}{Lemma}[section]
 \newtheorem{theorem}{Theorem}[section]
 \newtheorem{definition}{Definition}[section]
\fi

\usepackage{listings}

\usepackage[T1]{fontenc}
\usepackage{microtype}

\usepackage[scaled]{beramono}
\newcommand\Small{\fontsize{6.6pt}{6.8pt}\selectfont}
\newcommand*\LSTfont{\Small\ttfamily\SetTracking{encoding=*}{-60}\lsstyle}

\lstdefinestyle{numbers}{
         language=ML,
         basicstyle=\LSTfont,
         numbers=left,
         numbersep=5pt,
         extendedchars=true,
         breaklines=true,
         keywordstyle=\bfseries,
         morekeywords={mod,lap,expmech,expect,match,mlet,munit,cunit,clet,
         sample,random,def,function,for,expect_num,return,where},
         mathescape=true,
         literate={->}{{$\to$}}1%
                  {=>}{{$\Rightarrow$}}3%
                  {-->}{{$\qquad \leadsto$}}2,
         stringstyle=\ttfamily,
         showspaces=false,
         showtabs=false,
         xleftmargin=8pt,
         showstringspaces=false
 }

\lstdefinestyle{nonumbers}{
         language=ML,
         basicstyle=\LSTfont,
         numbers=none,
         extendedchars=true,
         breaklines=true,
         keywordstyle=\bfseries,
         morekeywords={mod,lap,expmech,expect,match,mlet,munit,cunit,clet,
         sample,random,def,function,for,expect_num,return},
         mathescape=true,
         literate={->}{{$\to$}}1%
                  {=>}{{$\Rightarrow$}}3%
                  {-->}{{$\qquad \leadsto$}}2,
         stringstyle=\ttfamily,
         showspaces=false,
         showtabs=false,
         showstringspaces=false
 }

\newcommand{\lstt}[1]{\mbox{\LSTfont #1}}

\definecolor{DarkGreen}{rgb}{0.1,0.5,0.1}
\definecolor{DarkRed}{rgb}{0.5,0.1,0.1}
\definecolor{DarkBlue}{rgb}{0.1,0.1,0.5}
\usepackage[]{hyperref}
\hypersetup{
    unicode=false,          
    pdftoolbar=true,        
    pdfmenubar=true,        
    pdffitwindow=false,      
    pdftitle={},    
    pdfauthor={}
    pdfsubject={},   
    pdfnewwindow=true,      
    pdfkeywords={keywords}, 
    colorlinks=true,       
    linkcolor=DarkRed,          
    citecolor=DarkGreen,        
    filecolor=DarkRed,      
    urlcolor=DarkBlue,          
}
\usepackage{cleveref}

\crefname{section}{\S}{\S}
\Crefname{section}{\S}{\S}

\def\lvmark{1}
\def\rvmark{2}

\renewcommand{\l}[1]{#1_\lvmark}
\renewcommand{\r}[1]{#1_\rvmark}

\usepackage{mathrsfs}

\newcommand{\List}[1]{\mathsf{list}\,{#1}}
\newcommand{\Distr}[1]{\mathsf{distr}\,{#1}}

\newcommand{\M}[2][]{%
            \ifthenelse{\isempty{#1}}%
              {\mathsf{M}\, {#2}}
              {\mathsf{M}_{#1}\, {#2}}}

\usepackage{tikz}
\usetikzlibrary{arrows}
\usepackage[numbers,longnamesfirst]{natbib}
\usepackage[subsectionbib]{bibunits}
\defaultbibliographystyle{abbrvnat}
\usepackage[ruled]{algorithm2e}

\SetAlFnt{\small}
\SetAlCapFnt{\small}
\SetAlCapNameFnt{\small}
\SetAlCapHSkip{0pt}
\IncMargin{-\parindent}

\DeclareMathOperator*{\argmax}{arg\,max}